\PassOptionsToPackage{x11names,table,dvipsnames}{xcolor}
\documentclass[acmtog,nonacm,balance=false]{acmart}

\acmSubmissionID{478}

\usepackage{booktabs} %

\usepackage[utf8]{inputenc}
\usepackage{amsfonts,amsmath,hyperref}
\usepackage{color}
\usepackage{tabularx}
\usepackage[most]{tcolorbox}
\usepackage{float}

\usepackage{enumitem}
\setlist[itemize]{leftmargin=*,noitemsep}

\usepackage{pifont}%

\usepackage{neural_subspaces} %

\newcommand{\PotentialEnergy}{E_\textrm{pot}}

\usepackage{xcolor}
\hypersetup{
    colorlinks,
    linkcolor={red!50!black},
    citecolor={blue!50!black},
    urlcolor={blue!80!black}
}

\usepackage[ruled]{algorithm2e} %

\SetAlFnt{\small}
\SetAlCapFnt{\small}
\SetAlCapNameFnt{\small}
\SetAlCapHSkip{0pt}

\copyrightyear{2023}
\acmYear{2023}
\setcopyright{acmlicensed}\acmConference[SIGGRAPH '23 Conference Proceedings]{Special Interest Group on Computer Graphics and Interactive Techniques Conference Conference Proceedings}{August 6--10, 2023}{Los Angeles, CA, USA}
\acmBooktitle{Special Interest Group on Computer Graphics and Interactive Techniques Conference Conference Proceedings (SIGGRAPH '23 Conference Proceedings), August 6--10, 2023, Los Angeles, CA, USA}
\acmPrice{15.00}
\acmDOI{10.1145/3588432.3591521}
\acmISBN{979-8-4007-0159-7/23/08}

\begin{document}

\title{Data-Free Learning of Reduced-Order Kinematics}

\author{Nicholas Sharp}
\affiliation{%
 \institution{NVIDIA, University of Toronto}
 \country{Canada}
}
\email{nsharp@nvidia.com}

\author{Cristian Romero}
\affiliation{%
 \institution{Universidad Rey Juan Carlos}
 \country{Spain}
}
\email{crisrom002@gmail.com}

\author{Alec Jacobson}
\affiliation{%
 \institution{University of Toronto, Adobe Research}
 \country{Canada}
}
\email{jacobson@cs.toronto.edu}

\author{Etienne Vouga}
\affiliation{%
 \institution{The University of Texas at Austin}
 \country{USA}
}
\email{evouga@cs.utexas.edu}

\author{Paul G. Kry}
\affiliation{%
  \institution{McGill University}
  \country{Canada}
}
\email{kry@cs.mcgill.ca}

\author{David I.W. Levin}
\affiliation{%
  \institution{University of Toronto, NVIDIA}
  \country{Canada}
}
\email{diwlevin@cs.toronto.edu}

\author{Justin Solomon}
\affiliation{%
  \institution{Massachusetts Institute of Technology}
  \country{USA}
}
\email{jsolomon@mit.edu}

\renewcommand\shortauthors{Sharp \etal{}}

\begin{abstract}
  Physical systems ranging from elastic bodies to kinematic linkages are defined on high-dimensional configuration spaces, yet their typical low-energy configurations are concentrated on much lower-dimensional subspaces.
  This work addresses the challenge of identifying such subspaces \emph{automatically}: given as input an energy function for a high-dimensional system, we produce a low-dimensional map whose image parameterizes a diverse yet low-energy submanifold of configurations.
  The only additional input needed is a single seed configuration for the system to initialize our procedure; no dataset of trajectories is required.
  We represent subspaces as neural networks that map a low-dimensional latent vector to the full configuration space, and propose a training scheme to fit network parameters to any system of interest.
  This formulation is effective across a very general range of physical systems; our experiments demonstrate not only nonlinear and very low-dimensional elastic body and cloth subspaces, but also more general systems like colliding rigid bodies and linkages.
  We briefly explore applications built on this formulation, including manipulation, latent interpolation, and sampling.
\end{abstract}

\begin{teaserfigure}
    \centering
    \includegraphics[width=\linewidth]{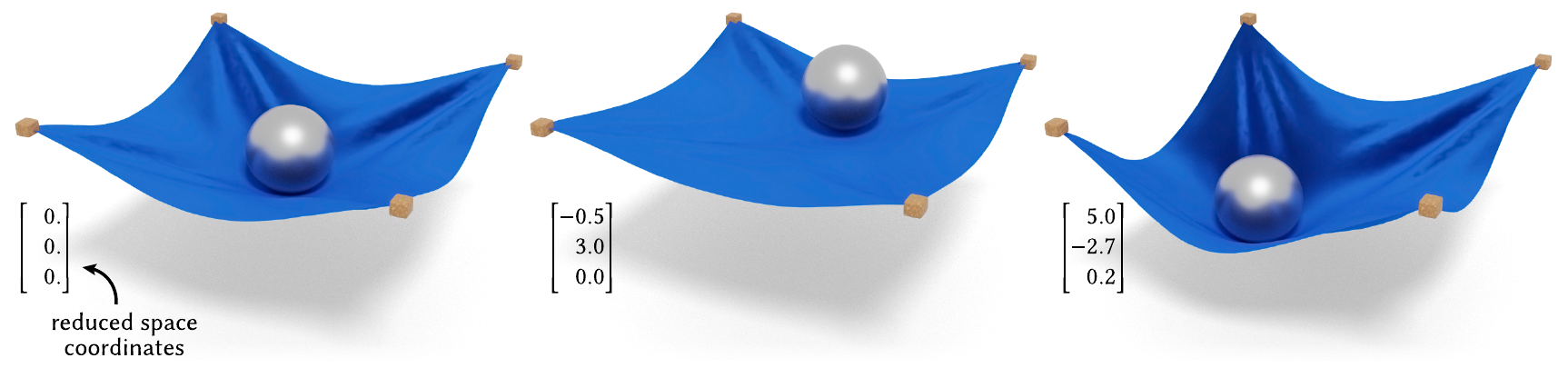}
    \vspace{-.3in}
    \caption{
      We fit neural networks as reduced-order models to learn low-dimensional approximations of complex physical systems. 
      This approach applies to a broad range of systems and requires no data as input, only a differentiable energy function and a seed state for sampling.
      Here, we show a 3-dimensional neural subspace for a system wherein a ball rolls around on a pinned cloth. 
      The ball location and cloth geometry are simultaneously encoded by the subspace, which is fit automatically from a potential energy that includes gravity, cloth bending and stretching, and a collision penalty between the ball and the cloth.
      \label{fig:teaser}
    }
    \vspace{1em} %
\end{teaserfigure}


\maketitle 

\section{Introduction}

Physical simulation algorithms perennially achieve new heights of detail and fidelity.  Modern computer graphics techniques successfully capture phenomena from elasticity to fluid motion, producing visual effects that are nearly indistinguishable from real life. With this added realism, however, comes substantial computational expense, often placing detailed physical simulation in the realm of offline computations involving many degrees of freedom.

In settings like interactive graphics, however, it is advantageous to reparameterize the system with a much smaller number of degrees of freedom which describe only the states that are actually of interest. 
These \emph{subspaces}, or \emph{reduced order models} enable downstream tasks, most traditionally fast simulation in reduced coordinates, but also other operations such as user-guided animation, interpolation, or sampling states of the system.

However, identifying such subspaces is inevitably challenging, because they must trade-off between the conciseness and expressivity.
Classical reduced-order simulation methods such as linear modal analysis or modal derivatives have typically focused on perturbative motions about a rest state for a deformable object.
These methods are highly effective numerical schemes for fast forward-integration of system dynamics; our approach will seek a complementary technique in two senses.

First, such approaches typically only approximate object behavior in a truncated region about the rest pose, and dramatic nonlinear motions are not well-represented in the subspace.
This concern is already impactful for the classic case of deformable bodies undergoing large motions, but is a total show-stopper when seeking reduced kinematics for more general physical systems, such as rigid bodies under collision penalties. 
In these settings, the kinematic landscape is so nonlinear that an approximation in terms of a local expansion does not capture any significant behavior.

Second, our subspaces will parameterize \emph{only} the desired configuration space of the system.
The challenge of large motions can be mitigated in perturbative methods by using a moderately large reduced basis.
However, this returns to the original problem with the full configuration space: the relevant system configurations again lie only on narrow submanifold of the space.
In contrast, we seek very low-dimensional but highly nonlinear subspaces, such that even large motions and physical systems with irregular potential landscapes can be directly parameterized; an important property for applications like animation and sampling.

This work is not the first to propose using a richer class of highly nonlinear models to fit reduced kinematics (see \eg{} \cite{Fulton2018,Holden2019,Shen2021,srinivasan2021learning}).
Our motivating goal is to do so \emph{without} data-driven fitting; we do not require any dataset of representative simulation trajectories or states as input. 
Collecting such a high-quality dataset is challenging and labor-intensive, both in the sense of engineering effort and user input.
It is a significant obstacle for past methods which otherwise offer excellent properties~\cite{hahn2014subspace,Fulton2018}.
To be clear, although our method leverages tools from machine learning, it is \emph{not} data-driven in the usual sense.
Instead, it mirrors recent ``overfit'' neural networks~\cite{xie2022neuralfields}, where models are fit in isolation to each example, and neural networks are used simply as a general and easy-to-optimize nonlinear function space.

\begin{figure}[b]
\begin{center}
  \includegraphics[width=\columnwidth]{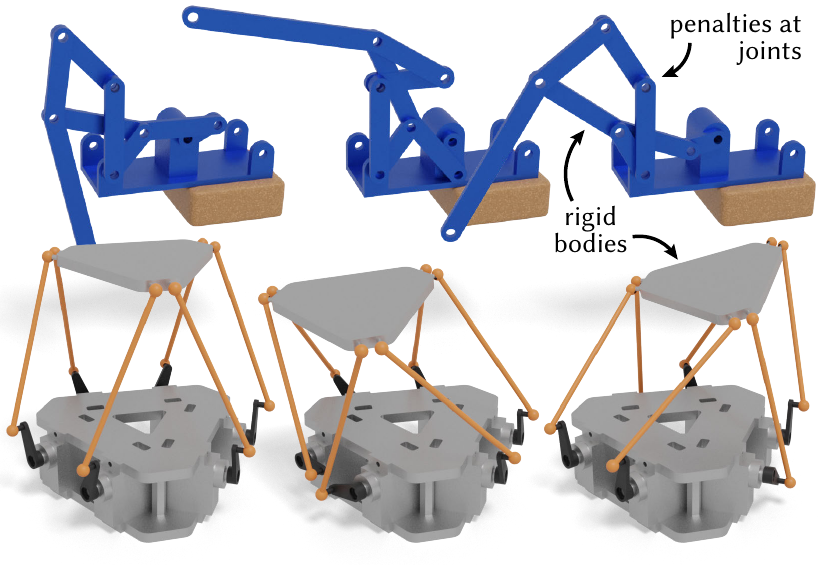}
    \caption{
      Our subspaces are effective on systems beyond deformable bodies and cloth.
      Here, linkages are naively described by the location of each segment, under a potential with strong penalties holding joints together.
      Neural subspace optimization automatically discovers low-dimensional kinematic motion as we fit a $1$d subspace to the Klann Linkage, \emph{top}, and a $3$d subspace to the Stewart Mechanism, \emph{bottom}.
      \label{fig:linkage}
    }
\end{center}
\end{figure}

\paragraph{Summary}
In this paper, we apply machine learning to identify a nonlinear reduced model for physical motion. Our approach is designed around two significant properties:
\begin{itemize}
    \item We do \emph{not} assume input data such as simulation trajectories are provided. Instead, our method is self-supervised, taking the energy function as input and automatically sampling it to explore the low-energy subspace.
    \item Our method is very general, and avoids specific assumptions about \eg{} deformable bodies. It applies broadly across systems such as rigid bodies and linkages under penalty potentials, or even multi-physics combinations of several different interacting systems.
\end{itemize}

Provided a differentiable potential energy function describing a given physical system and a single seed state from which to begin the search, our learning algorithm automatically determines an effective nonlinear low-order model, trading off between staying in low-potential energy configurations and coverage of the configuration space. 
Two loss parameters are exposed to adapt our objective to the system of interest.
Once fit, the parameterized kinematic subspace can be leveraged for a variety of purposes, from simulation via standard dense integrators in the subspace, to kinematic exploration and sampling.
We demonstrate our approach on a variety of physical systems, including deformable bodies, cloth simulations, rigid bodies under collision, and mechanical linkages.

\section{Related Work}

\paragraph{Subspaces for Simulation}
Subspace simulation methods have a long history in engineering and graphics, beginning with linear modal analysis~\cite{shabana1991theory,jamespai2002,vonTycowicz2015,hildebrandt2011interactive,james2006precomputed}. 
The survey of \citet{benner2015survey} provides a broad summary.
Linear modes provide a concise basis expressing deformation around an object's rest state. Fast simulation methods then restrict the equations of motion to this subspace. Difficulties arise in large-deformation settings wherein the basis size must be greatly increased to approximate nonlinearity~\cite{brandt2018hyper}. 
\citet{Barbic2005} augment the modal basis with second-order ``modal derivatives,'' while still resulting in a linear deformation subspace, and \citet{choi2005modal} and \citet{yang2015expediting} explore rotation and higher order terms.

While modal derivatives offset some disadvantages of linear modal analysis, both techniques are limited to representing deformations centered around a rest pose. This makes representing highly nonlinear deformations and effects difficult, and obstructs the application to more general physical systems as we show in~\figref{modes_comparison}.
Snapshot methods generalize beyond a region around the rest state by collecting large databases of simulation outputs and fitting a reduced space to that data. Initial algorithms used PCA~\cite{noor1980reduced} to construct an improved subspace, but the linear PCA basis still must be large to capture a wide range of deformations. 

Modern neural representations such as autoencoders~\cite{Fulton2018,Shen2021} offer a potential panacea. However, such methods again rely on simulation snapshots for training, and thus resort to user-guided sampling, making these methods time consuming and compute intensive. 
Like us, \citet{brandt2016geometric} sample configuration space, but do so in a way which still only interpolates specified configurations.
Even methods learning neural enrichments to linear subspaces~\cite{romero2021learning} suffer from the data generation problem; no successful self-supervised, data-free learning method for nonlinear subspaces has yet been demonstrated.

\paragraph{Neural and Data-Driven Methods}

Recent work across machine learning shows neural networks have significant potential to model complex physical systems efficiently~\cite{kochkov2021machine,tompson2017accelerating,pfaff2020learning,gao2021super}.
These approaches range from fitting update rules to observed data, to accelerating expensive numerical steps with data-driven proxies.
The most similar of these efforts tackle problems in dynamics and deformation, often with the goal of producing efficient real-time simulators~\cite{grzeszczuk1998neuroanimator,romero2020modeling,zheng2021deep}.
Applications of this work, as well as ours, include graphics, animation, robotics, design \cite{li2018learning,li2019propagation}.

The task of modeling dynamics and collisions in cloth has received particular attention~\cite{hahn2014subspace,Holden2019, bertiche2021pbns,zhang2021neuralGarments,santesteban2022snug, bertiche2022neural}.
In fact, \citet{bertiche2021pbns} and \citet{santesteban2022snug} leverage self-supervised setups which bear some similarity to ours, although many aspects of their approach are specific to garment modeling task. 
Additionally, a primary challenge in our setting is avoiding collapse of the subspace, while with clothing this is automatically handled by human body shape and motion distributions.

\section{Method}
\label{sec:Method}

We present a straightforward approach to fit a neural network modeling low-energy kinematics of a physical system. 
The formulation is general, applying to a broad set of systems and capturing both linear and nonlinear subspaces.
It follows widespread success fitting low-dimensional submanifolds of high-dimensional spaces using neural networks (\eg{} ~\cite{lee2020model}).

\subsection{Neural Subspace Maps}
\label{sec:NeuralSubspaceMaps}

Consider a map $f_\theta$, which takes a low-dimensional subspace $\R^d$ to the high-dimensional configuration space $\R^n$ of some physical system ($d\ll n$), so
  $f_\theta : \R^d \to \R^n.$ 
For example, $\R^n$ might represent the set of all possible vertex configurations for a given triangle mesh (so, $n=3|V|$ where $|V|$ is the number of vertices), while $\R^d$ parameterizes a space of deformations that move multiple vertices in tandem.  The vector $\theta\in\R^k$ contains learnable parameters specific to the physical system, e.g.\ neural network weights.

Classical simulation algorithms operate on $\R^n$, where the potential energy \smash{$\PotentialEnergy:\R^n\to\R$} and external forces can be evaluated directly; the expense of physical simulation then comes from the large number of variables $n$ that must be manipulated.
However, $\R^n$ contains many unlikely configurations, corresponding to high-energy deformations under the potential energy $\PotentialEnergy$.  
In many settings, we can reasonably expect the kinematics to stay in the image \smash{$f_\theta(\R^d)$} of some map $f_\theta$ parameterizing typical configurations.  

As a simple example, if we take $\theta=(A,x_0)$ for some $A\in\R^{n\times d}$ and $x_0\in\R^n$ with $f_\theta(z)=Az+x_0$, we recover the basic setup of linear modal analysis.  In this setting, $x_0$ is the rest state of the system, and the columns of $A$ parameterize low-energy perturbations of $x_0$. 
However, linear models cannot fit general nonlinear kinematics.

More broadly, efficient and accurate simulation demands a map $f_\theta(\cdot)$ spanning low-energy configurations of the system.  Unlike classical modal analysis, 
an immediate benefit of working with a more general $f_\theta$ is that $f_\theta$ can encode \emph{nonlinear} subspaces, rather than only linear modes.  In this work, we model $f_\theta$ as a neural network, with weights $\theta$ (see \secref{ImplementationDetails} for architectures).

\figref{rigidity_analysis} illustrates this property on a hanging stiff cow under elastic and gravitational potentials. Traditional modal analysis is limited to linear skewing about a rest pose, whereas our neural model finds a nonlinear subspace with curved swinging motion.

\begin{figure}
\begin{center}
    \includegraphics[width=\columnwidth]{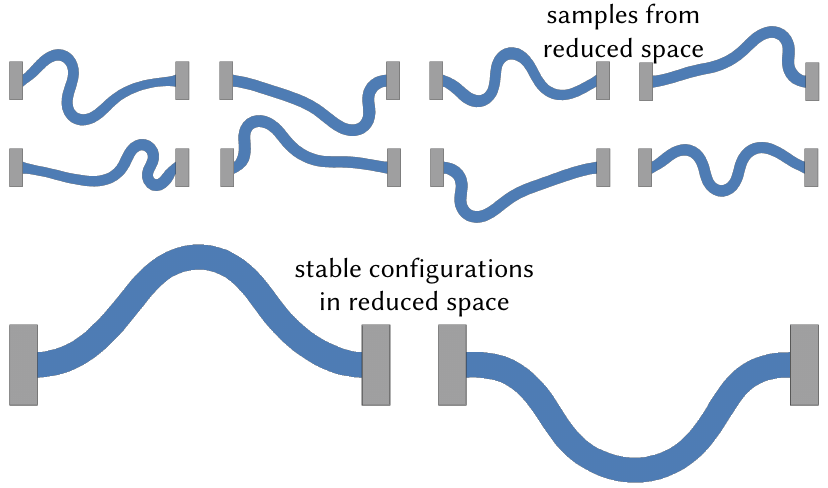}\vspace{-.15in}
    \caption{
      We fit an $8$-dimensional reduced space for a neohookean elastic bar under gravity in a compressed buckling configuration.
      \figloc{Top:} samples from the smooth yet irregular reduced space fit by the neural network.
      \figloc{Bottom:} the reduced space contains both of the stable buckled configurations.
      \label{fig:bistable_bar}
    }
\end{center}
\end{figure}

\subsection{Objective Function}
\label{sec:ObjectiveFunction}

For any choice of subspace architecture $f_\theta$, we will fit the parameters $\theta$ via stochastic gradient descent on an objective function.
We must then identify an objective function which drives $f_\theta$ to be a suitable, high-quality kinematic subspace.

Our key observation is that we can optimize for $\theta$ directly using the analytical description of the system---in particular, the potential energy function $\PotentialEnergy(\cdot)$---rather than requiring training data. 
This approach sidesteps the data collection needed for supervised models: We do not assemble a dataset of  motion trajectories or even forward-integrate the dynamics of the system during training.

We might seek low-energy subspaces of a system by minimizing the expected potential energy of randomly-sampled subspace configurations $z$ as follows:
\begin{equation}
\label{eq:subspace_loss_potential}
  \E_{z\sim \mathcal{N}} \bigg[\PotentialEnergy(f_\theta(z)) \bigg].
\end{equation}
Here, $\mathcal N$ denotes the Gaussian distribution over $\R^d$ with mean $0$ and variance $I_{d\times d}$; $\E$ denotes the expectation with respect to a random variable.  Note we have not put a scale on our latent variable $z$, so we are using $\mathcal N$ to capture a reasonable range of values. 

Minimizing \eqref{subspace_loss_potential} with respect to $\theta$, however, yields an uninteresting map $f_\theta$:  The optimal $f_\theta$ %
maps all latent variables $z$ to the lowest-energy configuration, i.e., the minimizer of $\PotentialEnergy$.

To combat the degeneracy above, we also expect our subspaces to span some sizable region of configuration space $\R^n$.  To accomplish this goal, we could attempt to impose isometry up to scale on $f_\theta$, e.g.\ by enforcing that
$$
|f_\theta(z) - f_\theta(z')|_M\approx \sigma|z - z'|
$$
for typical $z,z'\in\R^d.$ 
Here the distance in configuration space $\R^n$ is measured with respect to the system's mass matrix $M\in\R^{n\times n}$: $|x|_M^2 := x^\top M x.$
 Equivalently, we  write:
\begin{equation}\label{eq:loglipschitz}
\log\frac{|f_\theta(z) - f_\theta(z')|_M}{\sigma|z-z'|}\approx 0
\end{equation}

\clearpage
Enforcing strict equality in \eqref{loglipschitz} for all $z,z'$ is a stiff constraint; indeed, one can show that changing $\approx$ to $=$ above forces $f_\theta$ to be affine (see Proposition~\ref{prop:affine} in the supplemental material).  Hence, we instead use a soft penalty to avoid degeneracies:
\begin{equation}
\label{eq:subspace_loss_penalty}
  \E_{z,z' \sim \mathcal{N}} \left[ \left(\log \dfrac{|f_\theta(z) - f_\theta(z')|_M}{\sigma |z - z'|}\right)^2 \right] .
\end{equation}
Intuitively, this expression prefers maps $f_\theta(\cdot)$ whose Lipschitz constant is roughly $\sigma$ everwhere.
Similar formulations have recently been leveraged in other contexts by \citet{du2021learning}.

\begin{figure}[b]
\begin{center}
    \includegraphics[width=\columnwidth]{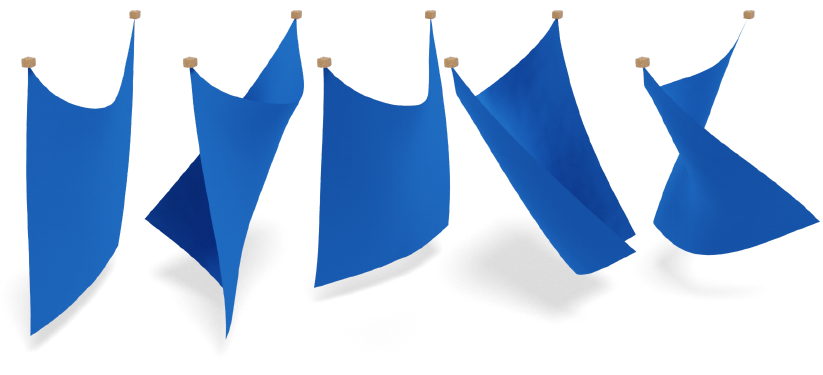}
    \caption{
      States from an $8$d reduced-order subspace of a thick hanging cloth.
      \label{fig:cloth_pinned}
      \vspace*{5em}
    }
\end{center}
\end{figure}

Combining these terms and using $\lambda\in\R$ as a weight, we optimize for the parameters $\theta$ as follows:
\begin{center}
\begin{tcolorbox}[ams equation,colback=white,colframe=black,arc=0mm,width=(\linewidth-1em),boxrule=0.2mm,left=0mm,right=1mm,top=0.5mm]
\label{eq:subspace_loss}
  \min_{\theta} \E_{z,z' \sim \mathcal{N}} \bigg[\PotentialEnergy(f_\theta(z)) 
                            + \lambda \left(\log \dfrac{|f_\theta(z) - f_\theta(z')|_M}{\sigma |z - z'|}\right)^2  \bigg]
\end{tcolorbox}
\end{center}
In this formulation, our latent subspaces map from a roughly unit-scaled region about the origin, with low-energy configurations concentrated near $0$.
To accelerate subspace fitting in practice, rather than sampling pairs $z,z'$ to evaluate ~\eqref{subspace_loss_penalty}, we estimate it pairwise on all $z$ samples in a training batch.

The hyperparameter $\sigma$ adjusts the size of the subspace. Small $\sigma$ yields subspaces that are tightly concentrated around low-energy configurations, while large $\sigma$ can force the map's image to include higher-energy states. 
The weighting parameter $\lambda$ should be chosen to ensure that \eqref{loglipschitz} roughly holds.
Hyperparameters inevitably arise due to the wide variety of scaling and units used in physical energies---\tabref{experiment_details} gives values for all experiments; see \secref{SelectingHyperparameters} for additional discussion.

\subsection{Reduction to Modal Analysis}\label{sec:modal}

When we take our parameters $\lambda$ and $\sigma$ to the extreme, our general model reduces to a classical linear method for modal analysis in physical simulation.  In particular, in Appendix~\ref{sec:reductionproof} we derive the following proposition:

\newpage

\begin{proposition}\label{prop:limit}
    Suppose $f_\theta$ has the capacity to represent affine functions.  Then, as $\sigma\to 0$ and $\lambda\to\infty$, the solution to \eqref{subspace_loss} satisfies
    \begin{equation}\label{eq:perturbativeformula}
\left\{
    \begin{array}{r@{\ }l}
        f(z) &= Az+b\\
        b &= \arg\min_b \PotentialEnergy(b)\\
        A &= \sigma \cdot \textsc{top-$d$-generalized-eigenvectors}(M,H(b)).
    \end{array}
    \right.
\end{equation}
\end{proposition}
\noindent In words, as we push to preserve geometry of the configuration space exactly ($\lambda\to\infty$) and to prioritize small neighborhoods ($\sigma\to0$), we recover a linearization about the minimum-energy state.

\subsection{Subspace Simulation}
\label{sec:SubspaceSimulation}
Although we focus primarily on simply encoding the kinematic subspace, if desired our subspaces $f_\theta$ can also be used for forward simulation via time integration in the latent space. 
In principle any integration scheme is compatible with our approach, we leverage a simple implicit Euler scheme.

We optimize to obtain the subspace configuration $\hat{z}$ in the next timestep~\cite{hahn2012rig} as:
\begin{equation}
\label{eq:TimestepOptimization}
    \hat{z} = \arg \min_{z} \left[ \frac{1}{2 h^{2}}\left| f_{\theta}(z) - \bar{q} \right|_{M}^{2} + \PotentialEnergy(f_{\theta}(z)) \right],
\end{equation}
where $h$ is the timestep and $\bar{q}$ is an inertial guess computed from previous configurations. 
The optimization is performed in the neural space $z$ via substitution into the optimization formulation~\cite{Fulton2018} and solved using L-BFGS~\cite{Liu89onthe}. 

The performance characteristics of this integration are very different from past methods; an advantage is that integration is performed in a small, dense space amenable to fast vectorized computation, while a disadvantage is that the nonlinearity of our subspaces may demand many optimization steps for accuracy.
In general such integration is significantly faster than simulation in the full space, but does not outperform existing specialized subspace integrators.
This work does not accelerate energy function evaluation, but could be used in conjunction with methods such as \citet{An2008}.

\subsection{Conditional Subspaces}

A benefit of our neural subspace formulation is that subspaces can easily be conditioned on auxiliary data such as material parameters and external constraints. Conditional parameters can be adjusted to adapt the subspace to a family of systems, and even can be varied dynamically at runtime.

More precisely, we can generalize $f_\theta$ to incorporate conditional parameters as additional inputs to the neural subspace map as
\begin{equation}
  \label{eq:conditional_subspace}
  f_\theta : \R^d \times \R^m \to \R^n \qquad q \gets f_\theta([z,c])
\end{equation}
where the conditional parameters are a vector $c \in \R^m$, and $[z,c]$ denotes vector concatenation. During training, we additionally sample from the space of system-defined valid conditional parameters to evaluate the expectation in \eqref{subspace_loss}.
In \figref{conditional_bar} we show an elastic bar conditioned on both the location of boundary conditions, and the material stiffness.

\begin{figure}[b]
\begin{center}
    \includegraphics[width=\columnwidth]{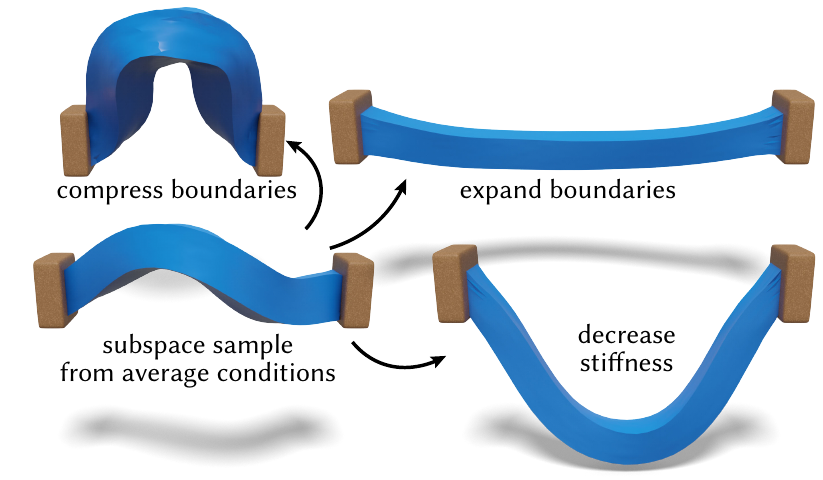}
    \vspace{-.2in}
    \caption{
      Neural subspaces are easily conditioned on parameters like boundary conditions or material properties.
      Here, we learn a subspace for an elastic bar, conditioned on the material stiffness ratio and locations of pinned endpoints.
      For each parameter choice, our subspace captures plausible kinematics.
      \label{fig:conditional_bar}
    }
\end{center}
\end{figure}

\vfill
\pagebreak

\section{Architectures and Training}
\label{sec:NetworkArchitectureAndTraining}

In principle, any neural architecture could be used to represent our subspace map $f_\theta$.
In this work we focus on the problem formulation itself, and consider only simple multi-layer perceptrons (MLPs).

\subsection{Seeded Subspace Exploration}
\label{sec:SubspaceExploration}

One important modification is needed to effectively train our subspaces, accounting for the difficulties of stiff physical energies, which are very different from typical data-driven machine learning losses.
The challenge is that when starting from a randomly-initialized neural subspace map, merely finding \emph{any} point on the low-energy submanifold in configuration space amounts to a surprisingly hard optimization problem.
Consider the case of an elastic body: samples from a randomly-initialized subspace network yield configurations with vertices randomly positioned in space, leading to extremely large energies and many inverted elements.

Although much recent work has tackled robust simulation of such systems in the classical setting~\cite{smith2018stableNeoHook,smith2019analytic,kim2020finite,lin2022isotropic}, here we have the added challenge that the system degrees of freedom are parameterized by a highly nonlinear neural network sampled stochastically at each optimization step.
As a simple solution, we propose a training procedure which explores the configuration space outward from an initial 
\emph{seed} configuration $q_\textrm{seed} \in \mathbb{R}^n$ provided as input to the method. 
Precisely, during training only we parameterize the neural subspace map as
\begin{equation}
  \label{eq:scheduled_neural}
  f_\theta(z) := \rho \textrm{MLP}_\theta(z) + (1-\rho) q_\textrm{seed}
\end{equation}
where $\rho$ is a scheduling parameter which linearly increases from $0 \to 1$ as training proceeds.
Crucially, at the conclusion of training, this seed state is entirely absent and the resulting network is an ordinary MLP.
Likewise, we also modulate the scale parameters $\sigma$ in \eqref{subspace_loss_penalty}, multiplying by a factor of $\rho$ because at initialization $f_\theta$ is a constant map to the seed state which cannot possibly achieve the target scale.

Intuitively, this training procedure ``grows'' the subspace outward as fitting proceeds, initially expanding about the seed state but ultimately gaining the freedom to parameterize an arbitrary map.
This approach does demand $q_\textrm{seed}$ as an additional input to the method, but in our systems this proved to be no additional burden: a suitable state was already implicit in the definition of the system, \eg{} the rest state of an elastic body, or the pinned-joint configuration of a linkage.
Also, we emphasize that the formulation in~\secref{Method} does not make any assumptions involving the seed state. It does not need to be a rest state or minimal-energy configuration, any somewhat low-energy initial state will do, and the resulting subspaces have little dependence on the choice of seed.
For example, in systems like \figref{bistable_bar} which have no single distinguished configuration, any choice will yield similar subspaces.

\begin{figure}
\begin{center}
    \includegraphics[width=\linewidth]{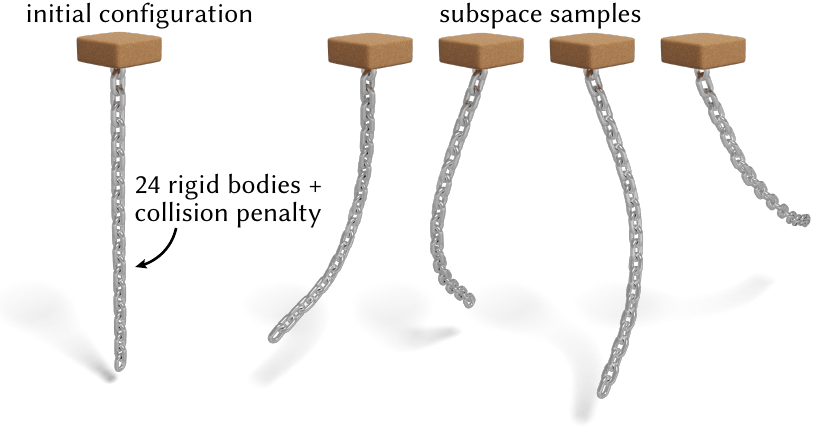}
    \caption{
      A chain of rigid bodies, held together by pairwise signed-distance collision evaluated at vertices.
      Despite the highly irregular potential landscape, our neural subspace fits an $8$d subspace spanning the large-scale continuum behaviors of the system.
      \label{fig:chain}
      \vspace*{-2em}
    }
\end{center}
\end{figure}

\begin{table*}
  \centering
  \caption{Parameters and dimensions for the experiments appearing in this work.
  All networks are MLPs, with the latent dimension, network size, and scaling parameter $\sigma$ adapted based on the degrees of freedom and desired range of motion for the kinematic space. See~\secref{SelectingHyperparameters} for extended discussion.
  \label{tab:experiment_details}}
  \rowcolors{2}{white}{CornflowerBlue!25}
  \begin{tabular}{l l r r r r r r r r r}
    \toprule
    \textbf{Name} & \textbf{Figure} & \textbf{Energy} & \textbf{Full dim} & \textbf{Reduced dim} & \textbf{Condition dim} & \textbf{MLP size} & $\lambda$ & $\sigma$ \\
    Cloth ball & \figref{teaser} & cloth + penalty              & 6069 & 3 & 0 & 5x128 & 1.0 & 0.05 \\
    Klann linkage & \figref{linkage} & rigid + penalty          & 84 & 1 & 0 & 5x64 & 1.0 & 1.0 \\
    Stewart platform & \figref{linkage} & rigid + penalty       & 168 & 6 & 0 & 5x128 & 0.5 & 1.0 \\
    Bistable bar & Figures~\ref{fig:bistable_bar},~\ref{fig:twisted_bistable} & neohookean           & 830 & 8 & 0 & 5x128 & 0.5 & $10^{-3}$ \\
    Hanging sheet & \figref{cloth_pinned} & cloth               & 6072 & 8 & 0 & 5x128 & 0.1 & 0.1 \\
    Conditional bar & \figref{conditional_bar} & neohookean     & 942 & 8 & 2 & 5x128 & 1.0 & 1.0 \\
    Chain & Figures~\ref{fig:chain},~\ref{fig:modes_comparison}  & rigid + collision & 288 & 8 & 0 & 5x128 & 1.0 & 1.0 \\
    Mini Chain& Figure~\ref{fig:chain_compare_autoenc}  & rigid + collision & 24 & 4 & 0 & 5x128 & 1.0 & 0.3 \\
    Hanging cow & \figref{rigidity_analysis} & neohookean       & 14676 & 3 & 0 & 5x128 & 0.5 & 0.1 \\
    Cantilever & \figref{modes_comparison} & neohookean         & 1500 & 2 & 0 & 5x128 & 0.5 & 0.5 \\
    Elephant & \figref{fulton_compare} & neohookean            & 1926 & 8 & 0 & 5x64 & 1.0 & $10^{-4}$ \\
    \bottomrule
  \end{tabular}
\end{table*}

\subsection{Implementation Details}
\label{sec:ImplementationDetails}

We implement all physical systems and neural networks in JAX~\cite{jax2018github}, leveraging automatic differentiation to compute derivatives.
Our neural networks use ELU activations and $5$ hidden layers.
The width of the hidden layers is adjusted from $64-128$ based on the scale of the problem.
\tabref{experiment_details} gives hyperparameters for all examples in this work, and \secref{SelectingHyperparameters} in the supplement provides further details about selecting parameters when applying our method to new and different physical systems.
An implementation is included in the supplement and at \url{github.com/nmwsharp/neural-physics-subspaces}.

We use the Adam optimizer~\cite{kingma2014adam} for training, and a learning rate of $10^{-4}$ for $10^6$ training steps, with batch size $32$. After each $250k$ training iterations the learning rate is decayed by a factor of $0.5$. 
Models are trained and evaluated on a single RTX 3090 GPU.
Memory usage is modest ($< 1$GB/model), and training times range from $1$ minute for small systems to $1$ hour for large systems.
Runtime performance when exploring or sampling the subspace is extremely fast; a single forward pass of our networks takes $<1ms$ and is dominated by pipeline latencies.
If the simulation is time-stepped in the subspace, performance is dominated by the cost of evaluating and differentiating the system's energy function, generally $10$s of milliseconds per timestep for the systems shown.
A scaling study of fitting and evaluation is included in the supplemental material.

\section{Evaluations}

\subsection{Physical Systems}

Here we summarize the systems/energies considered in this work. See \tabref{experiment_details} for problem sizes and parameter choices.

\paragraph{FEM} We use the finite element method (FEM) for the simulation and learning of deformable objects, discretizing the continuum in 2D and 3D examples using triangular and tetrahedral elements, respectively. We aggregate the contributions of all elements under a stable neo-Hookean material model~\cite{smith2018stableNeoHook} as the total potential energy.

\paragraph{Cloth model} We also model thin cloth sheets discretized as triangular surface meshes. 
Our experiments make use of a simple energy model with a bending term defined at edges~\cite{grinspun2003discrete}, and a constant-strain  Saint Venant–Kirchhoff (StVK) stretching term on faces.
In \figref{cloth_pinned} we compute a subspace for a pinned thick hanging sheet.
For now, we do not model cloth self-collisions, although \figref{teaser} shows basic cloth-object interactions (\figref{teaser}).

\paragraph{Penalty Functions}
\label{sec:PenaltyFunctions}

We can also include penalties to enforce constraints in the system. For instance, we use penalties to prevent collisions between objects and to enforce joint constraints in articulated mechanisms. Our penalty functions are defined as: 
\begin{equation}
    \label{eq:penalty}
    w_{eq} |C_{eq}(q)|^{2} + w_{ineq}  |\min(C_{ineq}(q),0)|^{2}
\end{equation}
where $C_{eq}, C_{ineq}$ are the equality and inequality constraint functions with the constraints to be imposed, and $w_{eq}, w_{ineq} \in \R^{+}$ are weighting factors.
These penalties must not go to infinity, because we must optimize through them during training even when sampling violates the constraint.
These penalties are added to the energy function, and otherwise our subspace fitting is applied as normal, fitting local constraints without any special treatment.
We demonstrate learning with penalty functions in Figures~\ref{fig:teaser} \&~\ref{fig:chain}, where inequality penalties and signed distance model collision, and \figref{linkage}, where an equality penalty holds linkage joints together.

\paragraph{Rigid Bodies}
We also consider rigid body dynamics, where each body's state is described by $\mathbb{R}^{12}$ unconstrained coefficients, interpreted as the entries of a $3 \times 4$ transformation matrix $\left[R|t\right]$.
An additional potential term $|(R^T R - I)|_2^2$ encourages the rotation component to be orthogonal via a Frobenius norm.
Collisions are implemented as naive penalties, testing all vertices of one shape against an analytical signed distance of the other.

\figref{chain} and \figref{modes_comparison} show a hanging chain modeled with 24 independent rigid bodies and signed distance function collision terms preventing separation between adjacent links.
This system is difficult to simulate classically, as even small-timestep implicit integrators get stuck on the energy landscape, yet the expected low-dimensional nonlinear continuum dynamics emerge automatically as we fit our subspace.

\paragraph{Mechanisms and Linkages} By combining rigid body dynamics with penalties holding joints together, we can also generate subspace models for complex linkages.
Reduced mechanism models have applications in engineering~\cite{lee2013proper,boukouvala2013reduced}.
As before, these linkages are naively modeled as free-floating rigid objects, with strong penalties at joints; no angular or relative parameterizations are used.
The learned subspace for the system automatically finds a low-dimensional parameterization for linkage motion.
\figref{linkage} demonstrates this behavior on the planar Klann linkage and 3D Stewart mechanisms, both of which have applications in robotics.

\begin{figure}
\begin{center}
    \includegraphics[width=\columnwidth]{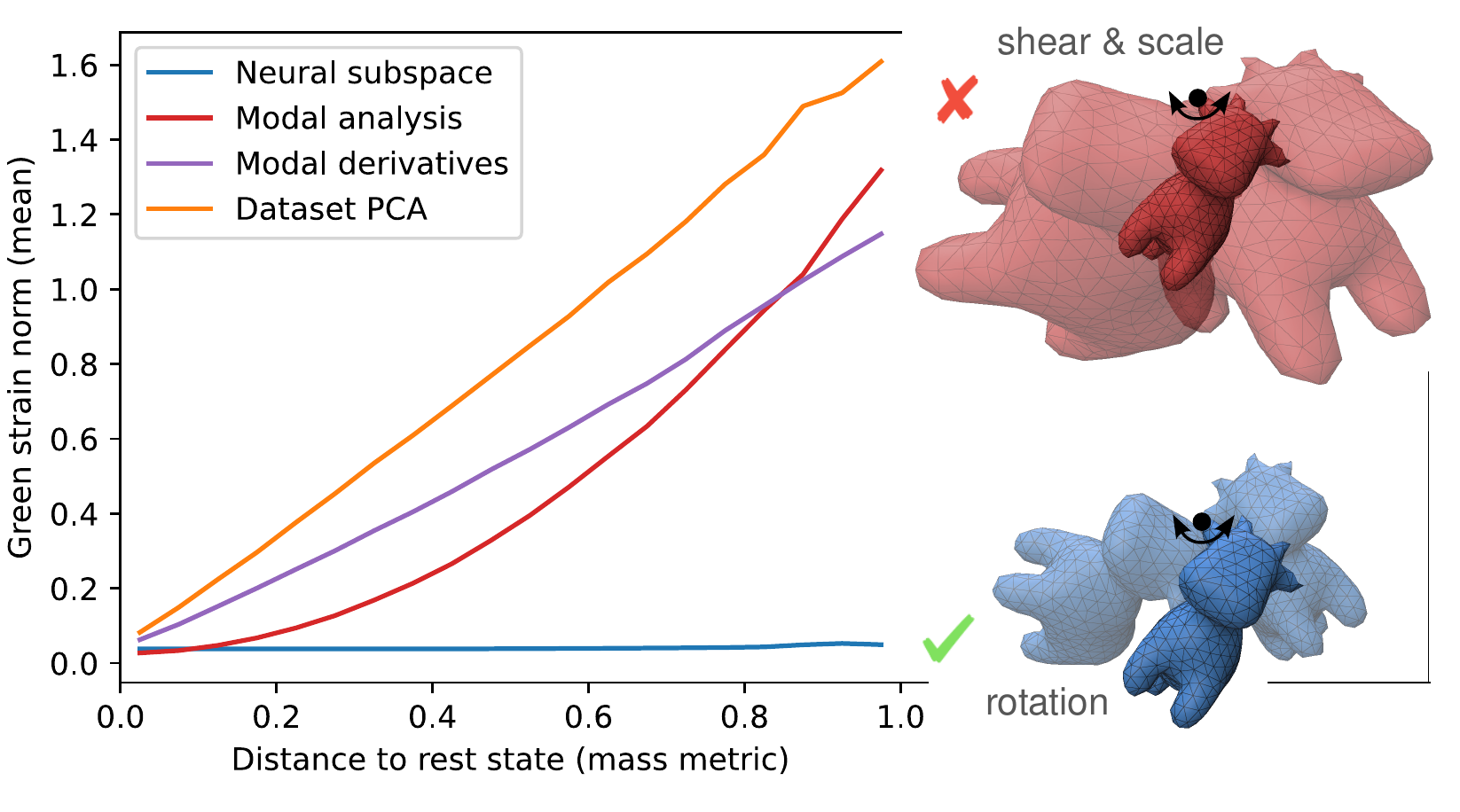}
    \caption{
      Our neural maps encode general nonlinear subspaces, greatly increasing expressivity at low dimension count.
      Here, we fit a 3d subspace to a stiff 3D deformable body under gravity, pinned at a single point, and measure the rigidity of the resulting subspace.
      Linear and quadratic approaches cannot encode the rotating motion in this low-dimensional subspace, whereas our method easily fits it.
      \label{fig:rigidity_analysis}
    }
\end{center}
\end{figure}

\begin{figure}
\begin{center}
    \includegraphics[width=\columnwidth]{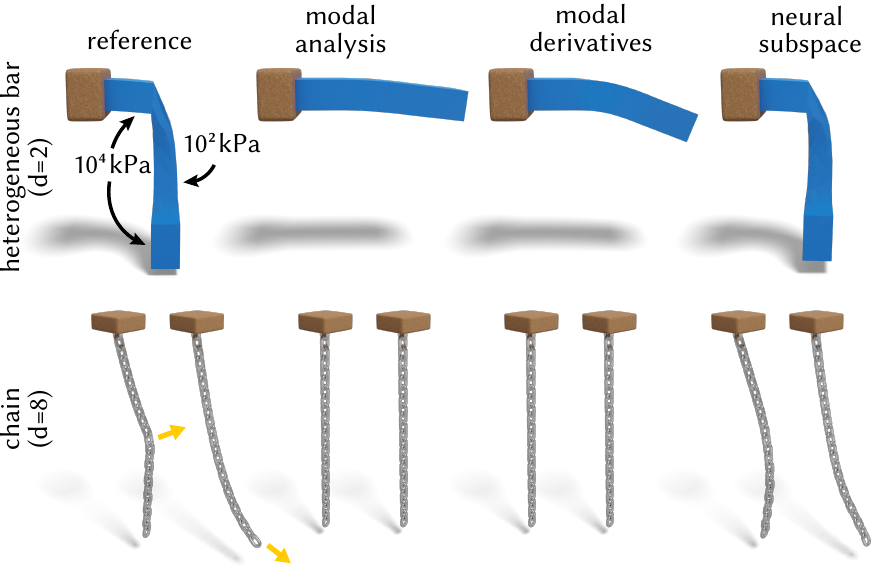}
    \caption{
      A comparison of data-free subspace generation methods: ordinary linear modal analysis, modal derivatives~\cite{Barbic2005}, and our approach, applied to a heterogeneous 3D bar and a rigid body chain with collision penalties.
      Each frame applies the same external load and visualizes the equilibrium response in the subspace.
      For the chain, the classic local methods contain no useful motions even at $d=8$.
      \label{fig:modes_comparison}
    }
\end{center}
\end{figure}

\subsection{Comparisons and Applications}
\label{sec:ComparisonsAndApplications}

In \figref{rigidity_analysis} we quantitatively evaluate $d=3$ subspaces computed on a stiff deformable object pinned at a single point, and measure the strain due to nonrigid deformation. We compare our method with linear modal analysis, modal derivatives and the principal component basis (PCA) of a small dataset recorded in an interactive session.
The linear approaches cannot represent rotations, and instead shear and scale the shape, while our approach fits a nonlinear rotation.
Note that recent supervised approaches have also tackled rotations by learning local motions coupled with a differentiable physics layer~\cite{srinivasan2021learning}.
Also, see~\secref{DataFreeVsSupervised} for an additional comparison to a baseline supervised approach.

In \figref{modes_comparison} we perform a side-by-side test of our subspaces, linear modal analysis, and modal derivatives on a heterogeneous deformable bar and rigid body chain.
For all methods we compute a subspace of matching dimension, then apply equivalent external loads and visualize the resulting subspace equilibrium.
In both experiments, our subspace much more-closely matches the expected reference physics, demonstrating the effectiveness of our approach for fitting low-dimensional yet expressive subspaces.

Recent work on generating rich nonlinear reduced spaces requires an input dataset of representative configurations or trajectories~\cite{Fulton2018,Holden2019,Shen2021}.
In these methods, data collection is a laborious, problem-specific process that requires humans in the loop or a scripted procedure to identify typical trajectories, as noted in \eg{} \cite[Sec 4.4]{Fulton2018} and \cite[Sec 6.1-6.2]{Shen2021}.
In contrast, our approach does not require an input dataset.
We generally do not expect our approach to outperform a supervised method trained on a sufficiently large and high-quality dataset; if a dataset is available it should certainly be used.

We also present two preliminary applications which show the complementary value of our low-dimensional data-free scheme.
Because our subspaces densely map on to the desired submanifold of configuration space, we can perform user-guided animation in the subspace.
The supplemental video shows a looping animation of the hanging chain constructed by choosing a set of keyframes in the latent space and applying a cyclic Catmull-Rom spline interpolation.
Additionally, in \figref{fulton_compare}, we use our automatic method as a \emph{sampler} for a downstream specialized supervised method, overcoming the primary limitation of needing to collect a dataset.

\section{Conclusion}

This work introduces a promising approach for fitting kinematic subspaces directly to physical systems, without gathering datasets of trajectories.

\begin{figure}
\begin{center}
    \includegraphics[width=\columnwidth]{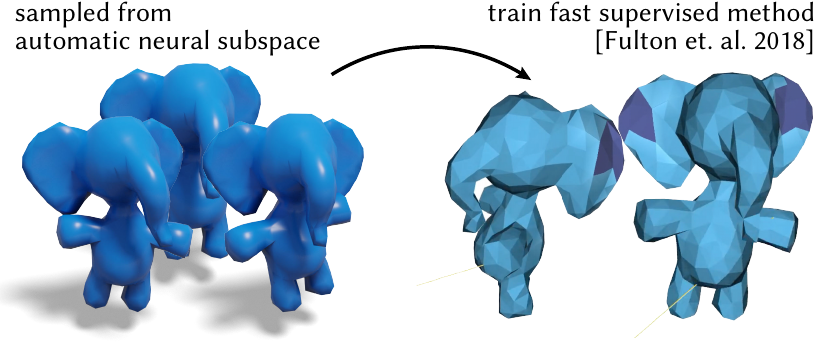}
    \caption{
      Our approach enables automatic sampling of diverse low-energy states of a system, useful for downstream data-driven applications.
      The AutoDef algorithm \cite{Fulton2018} offers fast, high-quality deformable simulation in reduced space, but requires a laborious process to collect training data.
      We first fit our subspace automatically to the deformable body of interest, then sample from the subspace to train~\cite{Fulton2018}, sidestepping the need for data collection.
      See supplement for details.
      \label{fig:fulton_compare}
    }
\end{center}
\end{figure}

\paragraph{Limitations}
Our subspace training procedure inherits both the difficulties of optimizing deep neural networks and of numerically integrating stiff physical systems.

\setlength{\columnsep}{1em}
\setlength{\intextsep}{0em}
\begin{wrapfigure}{r}{112pt}
  \includegraphics{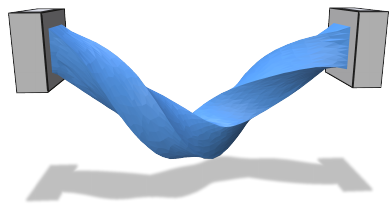}
      \caption{
      A twisted subspace.
      \label{fig:twisted_bistable}
    }
\end{wrapfigure}
Local minima may also result from isolated, locally stable configurations. For example, the inset figure shows an elastic bar pinned at both ends, where by-chance the training procedure has found a local minimum that respects the boundaries but has a $360^\circ$ turn.
More broadly, our subspaces may not perfectly reproduce desired configurations due to such local minima or simply insufficient model capacity, leading to artifacts that can be seen in some of our simulations (\eg{} the cloth in Figures~\ref{fig:teaser} \& ~\ref{fig:cloth_pinned}, and asymmetries in \figref{conditional_bar}).
We find that the training procedure in ~\secref{SubspaceExploration} generally avoids such artifacts, but it cannot guarantee to eliminate them.
Future work could develop numerical methods tailored to this hybrid problem.

Here we mainly consider low-dimensional subspaces, fitting higher dimensional subspaces $d > 10$ with our method does not necessarily capture additional effects, perhaps because \eqref{subspace_loss_penalty} becomes less effective as dimension increases.
More fundamentally, traditional subspaces methods are accompanied by theoretical analysis, while our neural networks currently have no corresponding physical theory to quantify which effects are captured and which are truncated, beyond the limiting models in \secref{modal}.

\paragraph{Future Work}
Here we used MLP architectures to encode the subspace map, but other network architectures could offer additional properties, such as equivariant networks~\cite{bronstein2021geometric} which build in rigid-invariance, or set-based networks~\cite{qi2017pointnet,wang2019dynamic} to model systems like particle fluids.

Future work could seek smoother subspaces, \eg{} via Lipschitz regularization of the MLP \cite{arjovsky2017wasserstein}, for faster implicit timestep convergence (\eqref{TimestepOptimization}) at large step sizes.
We could also expand our models for increased semantic control of subspaces, or to learn models that generalize over many physical systems at once.

\section{Acknowledgements}

The Bellairs Workshop on Computer Animation was instrumental in the conception of the research presented in this paper.

The authors acknowledge the support of the Natural Sciences and Engineering Research Council of Canada (NSERC), the Fields Institute for Mathematics, and the Vector Institute for AI.
This research is funded in part by the National Science Foundation (HCC-2212048), NSERC Discovery (RGPIN–2022–04680), the Ontario Early Research Award program, the Canada Research Chairs Program, a Sloan Research Fellowship, the DSI Catalyst Grant program and gifts by Adobe Systems.
The MIT Geometric Data Processing group acknowledges the generous support of Army Research Office grants W911NF2010168 and W911NF2110293, of Air Force Office of Scientific Research award FA9550-19-1-031, of National Science Foundation grants IIS-1838071 and CHS-1955697, from the CSAIL Systems that Learn program, from the MIT–IBM Watson AI Laboratory, from the Toyota–CSAIL Joint Research Center, from a gift from Adobe Systems, and from a Google Research Scholar award.

\bibliographystyle{ACM-Reference-Format}
\bibliography{bib}

\clearpage

\textbf{\Large Supplemental Material}

\appendix

\section{Derivation of Proposition~\ref{prop:limit}}\label{sec:reductionproof}

In the limit $\lambda \to \infty$, our metric-preserving constraint in \eqref{subspace_loss} is imposed precisely, forcing the subspace function to be linear and subject to mass orthonormality constraints (see Proposition~\ref{prop:affine} in Appendix~\ref{sec:uselesstheory}):
\begin{equation}
  \label{eq:linear_subspace_limit_lambda}
  f_{\theta}(z) = Az + b
  \quad \textrm{where} \quad
  A^T M A = \sigma^{2} I_{d\times d}.
\end{equation}
where $A\in\R^{n\times d}$, $b\in\R^n$, and $\theta=(A,b)$. 

If we take the limit of the scaling factor $\sigma \to 0^{+}$, the entries of the matrix $A$ become small due to the mass weighted orthonormality constraints in \eqref{linear_subspace_limit_lambda}. Hence, for sufficiently small $\sigma$, we can expect that the outputs of our subspace function $f$ in \eqref{linear_subspace_limit_lambda} will not stay far from the value of the inhomogeneous term $b$.

To understand behavior of our model in this perturbative regime, we can approximate our original formulation in \eqref{subspace_loss} using a truncated second order Taylor expansion of the potential energy $\PotentialEnergy(\cdot)$ centered at $b$:
\begin{equation}
  \label{eq:subspace_loss_quadratic}
  \min_{A,b} \E_{z\sim \mathcal{N}} \bigg[\PotentialEnergy(b) + g(b)^T A z + \frac{1}{2} z^T A^T H(b) A z \bigg] \\
  \quad \textrm{s.t.} \quad A^T M A = \sigma^{2} I.
\end{equation}
where $g(x)$ and $H(x)$ are the gradient and Hessian of the potential energy $\PotentialEnergy$, respectively.

Since $\mathcal N$ has mean zero, the linear term vanishes from the problem above.  Moreover, recognizing it as the Girard-Hutchinson trace estimator \cite{girard1987algorithme,hutchinson1989stochastic}, we can manipulate the quadratic term as follows:
$$
z^T A^T H(b) A z = tr(z^T A^T H(b) A z) = tr( A^T H(b) A zz^T),
$$
since a scalar is its own trace and by the property $tr(AB)=tr(BA)$.  Since $\mathcal N$ has the identity matrix as its covariance, we arrive at the following simplification of \eqref{subspace_loss_quadratic}:
\begin{equation}
  \label{eq:subspace_loss_quadratic_red}
  \min_{A,b} \bigg[\PotentialEnergy(b) + \frac{1}{2} tr \left( A^T H(b) A \right) \bigg] \\
  \quad \textrm{s.t.} \quad A^T M A = \sigma^{2} I.
\end{equation}
As $\sigma\to0$, the second term of the objective in \eqref{subspace_loss_quadratic_red} becomes negligible, and the potential energy term dominates.  Fixing $b$ and optimizing for $A$ yields a generalized eigenvalue problem.

Hence, we have motivated that as $\sigma\to0$ and $\lambda\to\infty$, we recover \eqref{perturbativeformula}. 
This formulation is exactly the classical linear method for modal analysis, discussed e.g.\ in \citet{shabana1991theory}.

\section{Affine Property}\label{sec:uselesstheory}

\begin{proposition}\label{prop:affine}
 Suppose $f(z):\Omega\to\R^n$ is $C^1$ on an open, connected domain $\Omega\subseteq\R^d$.  Then, $f(z)$ satisfies \eqref{loglipschitz} with equality for all $z,z'\in\Omega$ if and only if $f(z)=Az+b$ for some $A\in\R^{n\times d}, b\in \R^n$ with $A^T M A =\sigma^2 I_{d\times d}$.
 \end{proposition}
\begin{proof}
We start with the Lipschitz constant expression:
\begin{equation*}
    \left| f(z_{a}) - f(z_{b}) \right|_{M} = \sigma \left| z_{a} - z_{b} \right| \quad \forall z_{a},z_{b} \in \Omega,
\end{equation*}
where $\sigma$ is the positive Lipschitz constant. Squaring this expression implies
\begin{equation*}
    \left| f(z_{a}) - f(z_{b}) \right|^{2}_{M} = \sigma^{2} \left| z_{a} - z_{b} \right|^{2} \quad \forall z_{a},z_{b} \in \Omega.
\end{equation*} 
Since $f\in C^{1}(\Omega)$, taking the derivative w.r.t.\ $z_{a}$ implies
\begin{equation*}
 J(z_{a})^{T} M \left( f(z_{a}) - f(z_{b}) \right) = \sigma^{2} (z_{a} - z_{b}) \quad \quad J(z) := \left( \frac{\partial f}{\partial z}\right).
\end{equation*}
Taking the derivative of this expression w.r.t.\ $z_{b}$ implies
\begin{equation}\label{eq:jacinnerprod}
 J(z_{a})^{T} M J(z_{b}) = \sigma^{2} I_{d\times d} \quad \forall z_{a},z_{b} \in \Omega.
\end{equation}
Since the expression holds in the $z_a=z_b$ case, applying the expression above multiple times shows
\begin{equation*}
 \left( J(z_{a}) - J(z_{b}) \right)^{T} M \left( J(z_{a}) - J(z_{b}) \right) = 0.
\end{equation*}
Since $M\succeq0$, we thus have
\begin{equation*}
  J(z_{a}) = J(z_{b}) = \mathrm{const.} \quad \forall z_{a},z_{b} \in \Omega.
\end{equation*}
Since $f$ has a constant Jacobian in $\Omega$, it is automatically affine:
\begin{equation*}
    f(z) = A z + b \quad \forall z \in \Omega.
\end{equation*}
Moreover, $A^{T} M A = \sigma^{2} I_{d\times d}$ thanks to \eqref{jacinnerprod}.

To prove the reverse direction, note that the relationship
\begin{equation*}
    \left| A(z_{a} - z_{b}) \right|_{M} \quad s.t. \quad A^{T} M A = \sigma^{2} I_{d\times d}
\end{equation*}
implies
\begin{equation*}
    \sigma \left|z_{a} - z_{b} \right|.
\end{equation*}
\end{proof}

\section{Additional Details}
\label{sec:AdditionalDetails}

This sections gathers additional implementation and experimental details.

\subsection{Selecting Hyperparameters}
\label{sec:SelectingHyperparameters}

Like most learning-based approaches, our method requires a choice of hyperparameters to weight the objective function, in our case the penalty strength $\lambda$ and metric scaling parameter $\sigma$.
To be clear, adjusting two parameters is a modest burden as neural networks go, and our networks train very quickly (\secref{ImplementationDetails}); we discuss hyperparameter selection in-depth here to facilitate the application of our method to new physical systems.
Importantly, although hyperparameters may need new settings for new classes of systems (\eg{} cloth \vs{} kinematic mechanisms), they can be reused to fit many instances of a particular system (\eg{} many different cloth systems).
\tabref{experiment_details} gives hyperparameters for all examples in this work.

The hyperparameter $\sigma$ should be chosen based on the desired behavior of the subspace.
Large values yield a subspace which spans extreme states, while smaller values concentrate the subspace tightly around low-energy configurations.
The diversity of the subspace also affects the ideal neural network size.
Subspaces with large $\sigma$ which span a larger kinematic range may also require a larger network to accurately resolve the subspace, whereas smaller networks may be sufficient for a subspace with small $\sigma$ that only represents a narrow range of motions.

Note also that $\sigma$ depends on the physical units in which the configuration is measured.
We recommend initially choosing a large value for $\sigma$ and visualizing randomly-sampled system configurations during training, recalling that \secref{SubspaceExploration} linearly grows the subspace diversity as training proceeds.
For example, if the subspace spans a suitable range of configurations $1/3$ of the way through training, then $\sigma \gets 1/3 \sigma$ is a reasonable choice of parameter, and training can be repeated with this value.

The hyperparameter $\lambda$ weights the approximately-isometric objective; it should be chosen ensure \eqref{subspace_loss_penalty} has an effect, but also does not dominate the objective and enforce a restrictive affine subspace (see ~\secref{uselesstheory}).
This is easily assessed by measuring the unitless ratio in~\eqref{loglipschitz} during training, if is far from $1$ then $\lambda$ should be increased, and if it very close to $1$ (\eg{} within $10^{-3}$) then $\lambda$ should be reduced.

\subsection{Data-Free vs.\ Supervised Methods}
\label{sec:DataFreeVsSupervised}

\begin{figure}
\begin{center}
    \includegraphics[width=\columnwidth]{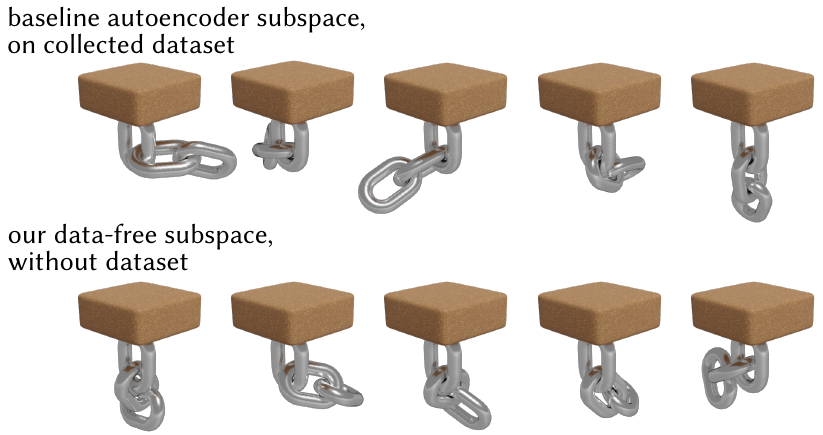}
    \caption{
      A comparison of samples from the latent space of an autoencoder trained on a manually-collected dataset (\emph{top row}), and our data-free approach on the same system (\emph{bottom row}). Both use the same latent dimension ($d=4$).
      \label{fig:chain_compare_autoenc}
    }
\end{center}
\end{figure}

The primary advantage of our formulation is that it fits a subspace using only potential energy function for a system, and does not require a training dataset.
Nonetheless, it is useful to consider how the quality of the subspaces compares with a baseline supervised method if a dataset were available.
To that end, we gather a dataset by interactively simulating the chain from \figref{chain}, here with fewer links to facilitate real-time robust full simulation.
The resulting dataset contains $40\textrm{k}$ sampled states of the system.
Our method is used to fit a subspace with a $5 \times 128$ MLP from a $d=4$ latent space, which does not require the dataset.
As a simple baseline model, we train an autoencoder, where the decoder is an MLP identical to our subspace model, and the encoder is a matching $5 \times 128$ hidden layer MLP.
The autoencoder is fit via reconstruction loss, along with a weak regularizer to encourage a $0$-centered latent space.
\figref{chain_compare_autoenc} shows samples from the resulting spaces.

\subsection{Performance Scaling}
\label{sec:PerformanceScaling}

\begin{figure}
\begin{center}
    \includegraphics[width=\columnwidth]{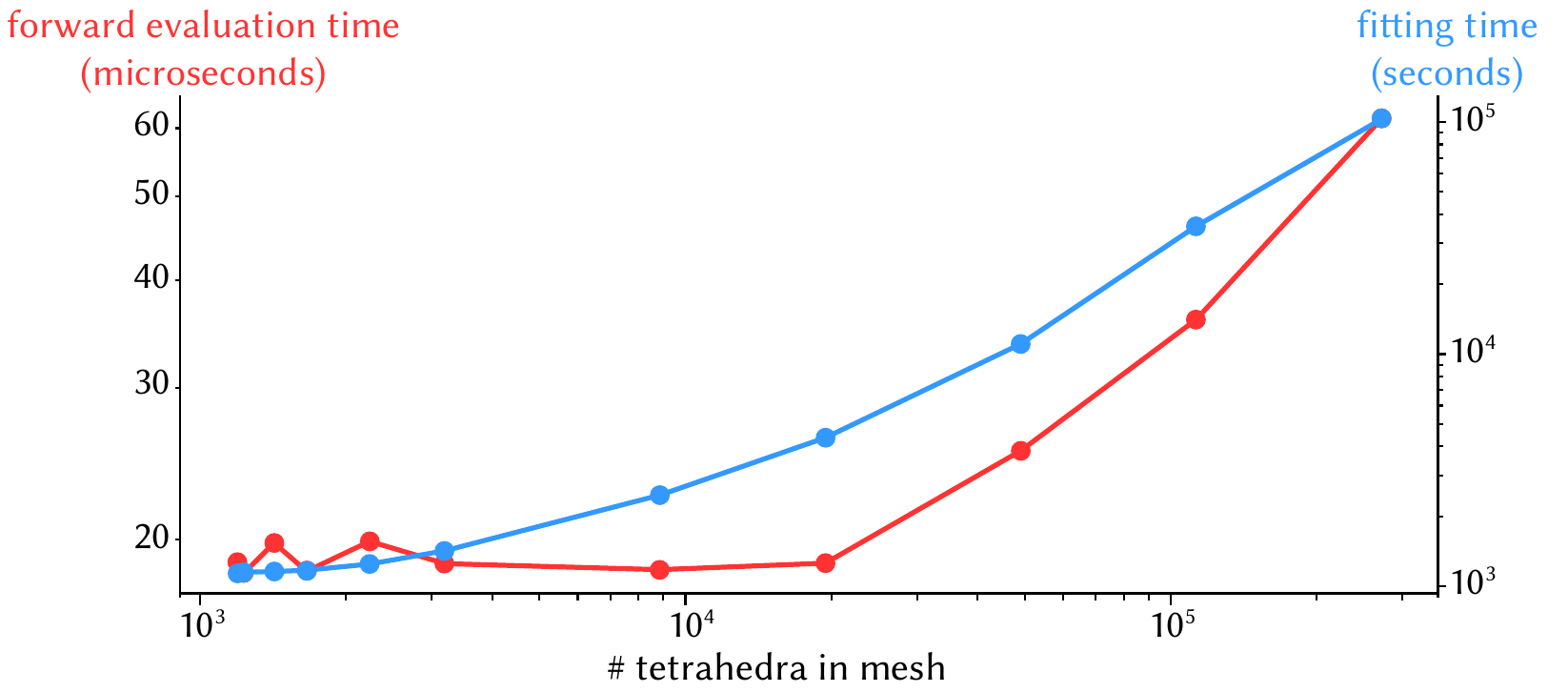}
    \caption{
      We evaluate the performance scaling of our method, fitting elastic deformation subspaces to the same object tetrahedralized at various resolutions. 
      Each data point is a fitted subspace at a different mesh resolution; the left red axis gives the time for a single forward evaluation of the  subspace map, while the right blue axis gives the fitting time.
      \label{fig:scaling_plot}
    }
\end{center}
\end{figure}

To measure the performance scaling of our method, we fit a series of subspaces to an elastic deformation system where the shape from \figref{rigidity_analysis} is discretized at various mesh resolutions ranging from $\approx 1\textrm{k}$ to $\approx 150\textrm{k}$ degrees of freedom.
\figref{scaling_plot} gives the corresponding time cost, measured on the same setup as in \secref{ImplementationDetails}.
Our method scales well to larger mesh sizes, especially for forward evaluation.

We naively use the exact same training scheme from \secref{ImplementationDetails} here and throughout this work. 
Likewise, all problem scales use the same $5 \times 128$ MLP model, increasing only the output dimension of the last layer to match the degrees of freedom for the system.
In practice one might adjust model sizes and training schedules for problems with vastly ranging orders of magnitude to tune performance.
Furthermore, approaches such as adaptive cubature~\cite{An2008} are well-suited to accelerate potential energy evaluation for high-resolution deformable models, which could greatly accelerate the fitting procedure.

\subsection{Experiment Details}
\label{sec:ExperimentDetails}

\paragraph{Sampling for \cite{Fulton2018}}

\secref{ComparisonsAndApplications} shows a preliminary application where our data-free subspace is used to sample data for an existing downstream supervised approach, sidestepping the need for dataset collection.
In particular, we automatically generate training data for the AutoDef method \cite{Fulton2018}, a recent approach which offers fast deformable simulations but requires significant effort to collect training data.
To do so, we first fit our subspace as usual to a single elephant mesh from the AutoDef experiment set, using the training parameters listed in ~\tabref{experiment_details}.
Then, we randomly sample 1000 simulation states by taking random sinusoidal motions in latent space, and applying our fitted subspace map $q \gets f_\theta(z)$ to get the corresponding system configurations. 
The states are encoded as displacements from the rest pose as expected by the AutoDef formulation.
These displacements are then used in-place of a manually collected training dataset to fit AutoDef as described in \citet{Fulton2018}.
The original AutoDef work proposes a nontrivial pipeline of user interaction to generate training data; we find that in this initial experiment substituting our automatically-sampled unsupervised population yields comparable results without the need to manually collect data.

\end{document}